\documentclass[conference,letterpaper]{IEEEtran}

\usepackage[utf8]{inputenc} 
\usepackage[T1]{fontenc}
\usepackage{url}
\usepackage{ifthen}
\usepackage{cite}
\usepackage[cmex10]{amsmath} 
\usepackage{bm}

\usepackage{relsize}
\usepackage[utf8]{inputenc}
\usepackage[english]{babel}
\usepackage{hyperref}
\usepackage{graphicx}

\usepackage{amsthm}
\usepackage{amssymb}
\usepackage{xcolor}
\usepackage[normalem]{ulem}

\newtheorem{theorem}{Theorem}

\newtheorem{lemma}[theorem]{Lemma}
\newtheorem{definition}[theorem]{Definition}

\newtheorem{example}[theorem]{Example}

\begin{document}
\title{Beyond Integral-Domain Stabilizer Codes} 


\author{%
  \IEEEauthorblockN{Lane G. Gunderman}
  \IEEEauthorblockA{Department of Electrical and Computer Engineering \\
                    University of Illinois Chicago\\
                   Chicago, Illinois, 60607\\
                    Email: lanegunderman@gmail.com}
}

\maketitle


\begin{abstract}
Quantum error-correcting codes aim to protect information in quantum systems to enable fault-tolerant quantum computations. The most prevalent method, stabilizer codes, has been well developed for many varieties of systems, however, largely the case of composite number of levels in the system has been avoided. This relative absence is due to the underlying ring theoretic tools required for analyzing such systems. Here we explore composite local-dimension quantum stabilizer codes, providing a pair of constructions for transforming known stabilizer codes into valid ones for composite local-dimensions. In addition remarks on logical encodings and the counts possible are discussed. This work lays out central methods for working with composite dimensional systems, enabling full use of the computational space of some systems, and expanding the understanding of the symplectic spaces involved.
\end{abstract}

Reliable quantum computation is expected to require quantum error-correction. The most promising approach is that of stabilizer codes--the quantum analog of classical additive codes \cite{gottesman1996class,gottesman1997stabilizer}. Traditionally stabilizer codes and quantum computations are performed on qubit systems where each qubit is comprised of a two-level system. There has been a fair amount of research done considering cases where the system has a prime number of levels, or prime-power number of levels, one survey includes \cite{ketkar2006nonbinary}. Likewise there has been a growing effort to examine bosonic (infinitely many leveled) encodings, as these encodings more closely match the nature of certain physical systems used in building quantum computers \cite{terhal2020towards,michael2016new,albert2022bosonic}. Unfortunately, the understanding of composite local-dimension systems (where there are a composite number of levels in the system), for instance a $^{173}\text{Yb}$ spin-$\frac{5}{2}$ device \cite{yang2024minute}, has been underdeveloped. In addition, in bosonic error-correction encodings which result in finite logical subspaces within the infinite system have recently been developed \cite{vuillot2024homological,xu2024letting}. These encodings utilize operations with finite order even within infinite spaces--such as operations acting as multiples of rotations by $\frac{\pi}{3}$, which has additive order of factors of $6$ while protecting against broader types of errors. For these reasons, this work develops constructions for composite local-dimension codes using existing stabilizer codes leveraging the local-dimension-invariant (LDI) framework, and connects with these bosonic encodings to provide a broader unifying framework to approach these. In short, this work explores systems whose local-dimension is best described by a ring that is not an integral-domain.

\section{Definitions and the local-dimension-invariant framework}



A commutative field, in this work denoted by $\mathcal{F}$, is equipped with two invertible operations: addition and multiplication. A commutative ring, in this work denoted by $\mathcal{R}$, is equipped with an invertible addition operation and a multiplication operation, but not all elements have multiplicative inverses. In this work we primarily focus on the situation where a field of fractions cannot be identified, and so extending beyond integral-domains, as the integral-domain case has already been studied \cite{gunderman2024stabilizer}. As this work extends to cases beyond qubits as well as beyond qudits (typically prime number of bases), we will use the more general term \textit{registers} in place of these terms \cite{watrous2018theory}, although as we will see later these registers may be hard to identify due to not being a discrete number of registers.

\subsection{Definitions for Stabilizer Codes}

In the discrete case with finite local-dimension $q$, the traditional qudit Pauli operators are generated by:
\begin{equation}
    X|j\rangle=|(j+1)\mod q\rangle,\quad Z|j\rangle=\omega_q^j |j\rangle,\quad \omega_q=e^{2\pi i/q}.
\end{equation}
We denote by $\mathbb{P}_q$ the single qudit Pauli operator group generated by the powers of the operators $X$ and $Z$, while $\mathbb{P}_q^n$ indicates a tensor product of $n$ such operators. In this work we will be varying the local-dimension.

\begin{definition}
A stabilizer code, specified by its $n-k$ commuting generators acting on $n$ registers, is characterized by the following set of parameters:
\begin{itemize}
\item $n$: the number of (physical) registers that are used to protect the information.
\item $k$: the number of encoded (logical) registers.
\item $d$: the distance of the code, given by the lowest weight of an undetectable generalized Pauli error. An undetectable generalized Pauli error is an $n$-qudit Pauli operator which commutes with all elements of the stabilizer group, but is not in the group itself.
\end{itemize}
These values are specified for a particular code as $[[n,k,d]]_q$, where $q$ is the local-dimension of the registers.
\end{definition}

In some places the parameter $K$ is discussed in place of $k$ which is given by the total number of orthonormal bases for the logical subspace.

Working with tensors of Pauli operators can be challenging, and so we make use of the following well-known mapping from these to vectors in symplectic spaces, following the notation from \cite{gunderman2020local}. This representation is often times called the symplectic representation for the operators \cite{lidar2013quantum,ketkar2006nonbinary}, but we use this notation instead to allow for greater flexibility, particularly in specifying the local-dimension of the mapping. This symplectic linear algebraic representation will be used for our proofs.

\begin{definition}[$\phi$ representation of a qudit operator]
We define the linear surjective map: 
\begin{equation}
\phi_q: \mathbb{P}_q^n\mapsto \mathbb{Z}_q^{2n}
\end{equation}
which carries an $n$-register Pauli in $\mathbb{P}_q^n$ to a $2n$-entried symplectic vector of integers mod $q$, where we define this mapping by:
\begin{equation}
I^{\otimes i-1} X^a Z^b I^{\otimes n-i} \mapsto \left( 0^{i-1}\ a\ 0^{n-i} \middle\vert 0^{i-1}\ b\ 0^{n-i}\right),
\end{equation}
which puts the power of the $i$-th $X$ operator in the $i$-th position and the power of the $i$-th $Z$ operator in the $(n+i)$-th position of the output vector. This mapping is defined as a homomorphism with: $\phi_q(s_1 s_2)=\phi_q(s_1)\oplus \phi_q(s_2)$, where $\oplus$ is component-wise addition mod $q$. We denote the first half of the vector as $\phi_{q,x}$ and the second half as $\phi_{q,z}$.
\end{definition}



We may invert the map to return to the original $n$-register Pauli operator with the phase of each operator being undetermined. We make note of a special case of the $\phi$ representation:

\begin{definition}
Let $q$ be the dimension of the initial system. Then we denote by $\phi_\infty$ the mapping:
\begin{equation}
    \phi_\infty:  \mathbb{P}_\infty^n\mapsto \mathbb{Z}^{2n}
\end{equation}
where operations are no longer  taken modulo some base, but instead carried over the full set of integers.
\end{definition}

$\phi_\infty$ allows us to avoid being dependent on the local-dimension of our system when working with our code. The ability to still define $\phi_\infty$ as a homomorphism, upon quotienting out the leading phase of the operators, (and with the same mapping) is a portion of the results of \cite{gunderman2020local,gunderman2024stabilizer}. 

Determining whether two operators commute in this picture is given by the following definition:
\begin{definition}
Let $s_i,s_j$ be two qudit Pauli operators over $q$ bases, then these commute if and only if:
\begin{equation}
\phi_q(s_i)\odot \phi_q(s_j)\equiv 0 \pmod q
\end{equation}
where $\odot$ is the symplectic product, defined by:
\begin{multline}
\phi_q(s_i)\odot \phi_q(s_j) \\ =\oplus_{k=1}^n [\phi_{q,z}(s_j)_k\cdot  \phi_{q,x}(s_i)_k- \phi_{q,x}(s_j)_k \cdot \phi_{q,z}(s_i)_k]
\end{multline}
where $\cdot$ is standard integer multiplication $\mod q$ and $\oplus$ is addition $\mod q$.
\end{definition}


In this work, we will assume that a stabilizer code is prepared in canonical form, given by $[I_{n-k}\ X_2\ |\ Z_1\ Z_2]$, where $X_2$ is a $(n-k)\times k$ block from putting the code into this form, while $Z_1$ is a $(n-k)\times (n-k)$ block and $Z_2$ is a $(n-k)\times k$ block also resulting from putting $\mathcal{S}$ into canonical form. So long as the local-dimension is a prime number, this can be done, although not uniquely per se \cite{gottesman1997stabilizer}. These form our basic tools for analyzing stabilizer codes.

\subsection{Local-dimension-invariant Codes}

Before delving into new results we will provide a brief summary of prior results for our primary tool of local-dimension-invariant stabilizer codes as introduced in \cite{gunderman2020local}. 

\begin{definition}
A stabilizer code $\mathcal{S}$ is called \textbf{local-dimension-invariant} iff:
\begin{equation}
    \phi_\infty(s_i)\odot \phi_\infty(s_j)=0,\quad \forall s_i, s_j\in \mathcal{S}
\end{equation}
\end{definition}
Importantly this requires \textit{exact} orthogonality of the vectors under the symplectic inner product. The following theorem from \cite{gunderman2024stabilizer} provided a prescriptive method for transforming any given stabilizer code defined over a finite field into an equivalent set of generators that define a stabilizer code for any choice of a finite, commutative ring:

\begin{theorem}\label{fldi}
Let $\mathcal{S}$ be a stabilizer code with parameters $[[n,k]]_{\mathcal{F}}$ for local-dimension given by finite field $\mathcal{F}$, with $r$ linearly independent generators in the $\phi$ representation. Then there is a prescriptive method to transform $\mathcal{S}$ into a local-dimension-invariant form still with parameters $[[n,n-r]]_{\mathcal{R}}$ for any choice of a finite, commutative ring $\mathcal{R}$.
\end{theorem}

We will delve into this prescriptive method later, as this provides our first method for constructing composite local-dimension codes. The prior only ensures that there are generators for the code over these local-dimensions. In order to ensure that the distance of the quantum error-correcting code of the LDI representation is at least as large as that of the original stabilizer code, we recall the theorem from \cite{gunderman2020local}:
\begin{theorem}\label{ogproof}
For all primes $p>p^*$, with $p^*$ a cutoff value greater than $q$, the distance of an LDI form for a stabilizer code $[[n,k,d]]_q$ used for a system with $p$ bases, $[[n,k,d']]_p$, has $d'\geq d$.
\end{theorem}

This only holds for prime values, however, it will form a crucial element in our distance preservation condition in the composite case. For this result about the distance of the code one breaks down the set of undetectable errors into two sets. These definitions highlight the subtle possibility of the distance reducing upon changing the local-dimension and will be leveraged later in the primary theorem's proof.

\begin{definition}\label{unavoidable}
An \textbf{unavoidable error} is an error that commutes with all stabilizers and produces the $\vec{0}$ syndrome over the integers.
\end{definition}


The other possible kind of undetectable error for a given number of bases corresponds to the case where some syndromes are multiples of the local-dimension:

\begin{definition}\label{artifact}
An \textbf{artifact error} is an error that commutes with all stabilizers but produces at least one syndrome that is only zero modulo the base.
\end{definition}

With these tools we can now show our constructions and remarks on features of composite local-dimension codes. Importantly, the LDI framework is not required to generate composite codes, however, it provides initial constructions.




\section{Composites, Factors of the Reals Modulo a Value, and Integer Multiples}


We now turn to using these definitions to provide construction methods for codes operating over more general commutative rings. We focus here on cases of more physical importance: systems that have a composite number of levels, such as $6$-levels, as well as checks which correspond to multiples of a rational fraction of $2\pi$\footnote{For instance, multiples of rotations by $\frac{\pi}{3}$ is isomorphic to $\mathbb{Z}_6$, and generally $\frac{2\pi}{r}$ for some integer $r$ is isomorphic to $\mathbb{Z}_r$, although can correct for errors within $\mathbb{R}_{2\pi/r}$.}, and systems operating within fixed multiplicity subspaces, such as only occupying excitations which are a multiple of $3$ ($3\mathbb{Z}$). We will first focus on the composite case.

The perhaps simplest way to convert a known stabilizer code into a code with a composite local-dimension is to directly apply the prescriptive method from Theorem \ref{fldi}. This Theorem is accomplished by putting the code into canonical form ($\begin{bmatrix} I_r & X_2&|\ & Z_1 & Z_2 \end{bmatrix}$) and adding to the $Z_1$ block the lower-triangular matrix of symplectic products between generators, taken over the integers. Manifestly each row has additive order of the local-dimension in this case. Further, the number of generators is unchanged from the initial code so it transforms an $[[n,k]]_q$ code into an $[[n,k]]_Q$ code with $Q$ being a chosen composite number. This serves as our first construction.

The following lemma provides a method for determining the dimension of the logical subspace.


\begin{lemma}\label{ldim}
Let $Q$ be a composite, square-free number. Then the logical subspace, $K$, has dimension $Q^n/\prod_{i=1}^m p_i^{\bar{k}_i}$, where there are $\bar{k}_i$ many independent generators with additive order $p_i$.
\end{lemma}

This result follows from the standard proof of the size of the logical subspace whereby the dimension is computed by tracing over the codespace projector, with the only difference here being that some generators may have different order. Of particular note in the above is that $K$ will always be a positive integer, while the number of logical registers $k=\log_Q(K)$ will not generally be. One can alter the code so that $k$ is a non-negative integer by appending logical operators of the different orders to the generators so that the denominator becomes $\prod_{i=1}^m p_i^{\max_i \bar{k}_i}$, effectively gauging out the extra logical bases so that the information is encoded within $Q$-level registers exclusively.

Now that we have a construction method and an expression for the size of the logical subspace, we now provide conditions on the distance of the code based on the distance of the original code. The following Theorem provides such a sufficient condition for ensuring the distance is at least kept.


\begin{theorem}[Composite case]\label{pdist}
Let $\mathcal{S}_\infty$ be an $[[n,k,d]]_q$ that has been put into an LDI form. Let $B$ be the largest (in absolute value) entry in $\mathcal{S}_\infty$. We may transform an $[[n,k,d]]_q$ code into an $[[n,k,d']]_Q$ code, with $d'\geq d$, in the following cases: i) for $q\mid Q$ and each prime factor of $(Q/q)$ exceeding $B$, or ii) each prime factor of $Q$ is greater than a cutoff value $p^*$.
\end{theorem}

Before proving this result, it has been shown before that for prime values such a $p^*$ exists and various bounds have been provided for it and $B$--both being functions of $k$, $d$, and $q$ for the original code \cite{gunderman2020local,gunderman2022degenerate}. $p^*$ generally grows rapidly, however, some cases where it remains at the initial local-dimension value have been shown, and $p^*$ is merely a bound not the value truly required.

\if{false}
\begin{proof}
Denote by $\phi_q(C)$ the symplectic representation for the code over local-dimension $q$. Let the composite local-dimension be given by $Q=qr$ with $q$ being the original local-dimension and $r$ being a different prime. Further let $r>p^*$ (the value for the prime case), for which $rank_r(M)=rank_q(M)$ for all matrix minors $M\subset \phi_q(C)$. For those minors that are full rank it means that the columns are linearly independent over $GF(q)$ and also $GF(r)$--while those that were not full rank, the size of the minimally independent set of columns match. And more generally considering the rank over integers this means that the entries of a linear combination cannot all be a multiple of their respective base--at least one entry is nonzero in the base.

Further, let us assume $r>B$ too. Assume a full rank submodule, then the linear combination of columns for a given submodule will then have its entries upperbounded by $(r-1)(q-1)+q\cdot \frac{B}{q}$, which means $[(r-1)-(q-1)]+(r-1)\cdot B$ $\leq r-1+B(r-1)$, which in turn is bounded by $r-1-B$, so having $-B+ r-1\neq qr$ as the first term is $-q\kappa$ which must be $\kappa<r$ which is true since $\kappa<B<r$. Then $r>B/q$ suffices which is more than covered when $r>p^*$.

For any submodule (when it was a matrix) full rank a priori adds $q[\leq \frac{B}{q}]$.
    Direct composite case. Doing a $p^*$ like argument but also using column l.c. since modules  

Sum of $\partial$ columns is given by $\partial (^< q+B)(r-1)\equiv ((r-1) ^< q -B)d$, which assuming $d<r$ and $d\neq q$, and using $r>B$. This expression cannot be a multiple of $qr$ as $(r-1)[^<q+B/r]$ has $|B/r|<1$ and $|q-^< q|\geq 1$ so the product cannot be congruent to $qr$ therefore at least one entry in the linear combination is still nonzero and the distance is preserved.

Still need to cover $\partial=q$ case, which likely is ok since $(r-1)$ factor.

A fair bit more polishing is needed, but I think roughly this argument holds.

$r>2d(q+B)/q$ likely just gives it, along with $r>p^*$. Need to show full rank basically for $q$ and $r$ cases. $(B/q)*q*(r-1)$
\end{proof}
\fi

\begin{proof}
Let $\phi_\infty(\mathcal{S}_\infty)$ be the symplectic representation of the LDI form used for the stabilizer code $\mathcal{S}$ over $q$ levels. Now we consider the possible impact on the distance upon selecting a new local-dimension $Q$. Note that in all cases the set of unavoidable errors will not change, while it is possible to introduce artifact errors which reduce the distance.

First, let $q\mid Q$ and all prime factors of $Q/q$ be larger than the largest (in absolute value) entry in $\phi_\infty(\mathcal{S}_\infty)$, denoted by $B$. Let $m$ be a matrix minor from $\phi_q(\mathcal{S})$. The distance of the code, restrained by the kernel of the code, can decrease due to the rank of $m$ decreasing upon selecting a new local-dimension. If $m$ is already not full rank it corresponds to either an unavoidable error or an artifact error due to the initial local-dimension $q$. If, however, $m$ is full rank, this means that all columns are linearly independent and so there is no way to take a linear combination of columns, over the integers, such that all columns are a nonzero multiple of $q$. Given this, when considering the underlying ring as over $\mathbb{Z}_Q$, with $q\mid Q$, it will still be impossible to take a linear combination of columns and obtain nonzero multiples of $Q$ in the syndrome. Lastly, we must avoid inadvertently altering the additive order of the parity checks, as otherwise the logical subspace's dimension will reduce. This is ensured so long as all prime factors of $Q/q$ are larger than $B$ so that no entry in $\phi_\infty(\mathcal{S}_\infty)$ can be a divisor of $Q$ (the entries will vary over $\mathbb{Z}_Q$ from $0$ to $B$ and $Q-B$ to $Q-1$, and hence are coprime with $Q/q$).

Second, let $q\nmid Q$ and all prime factors of $Q$ are larger than a cutoff value $p^*$, which is a function of $k$, $d$, and $q$ from the original code. In this case since each prime factor of $Q$ is larger than $p^*$ the rank of all matrix minors $m$ from $\phi_q(\mathcal{S}_\infty)$ when considered as over each prime factor's local-dimension, $\phi_p(\mathcal{S}_\infty)$ for $p\mid (Q/q)$, will not have a diminished rank. Therefore, by the same argument as above, the module corresponding to $m$ will also not have a diminished number of independent columns and thus cannot have the distance decreased.
\end{proof}
\if{false}
\begin{proof}
Let $\phi_q(C)$ be the symplectic representation for the code over local-dimension $q$. Further, let $Q=pq$ with $p$ and $q$ primes, with $p>p^*$ so that for all matrix minors $m$ from $\phi_\infty(C)$ the rank of $m$ modulo $q$ and modulo $p$ are always equal. Another interpretation of this condition is that if there are no linear combination of columns of $m$ which are congruent to $0$ modulo $q$, then there are also no linear combination of columns from $m$ which are congruent to $0$ modulo $p$.

Now we consider the case where the code's generators, in the symplectic representation, are over the ring $\mathbb{Z}_Q$. The entries in the symplectic representation are unchanged, however, the coefficients for linear combination of columns has changed to a ring--and thus a module. Let us consider a submodule such that when the operations are over the field $\mathbb{Z}_q$ it has full rank. Each such submodule has no set of columns which add to a vector which is precisely a multiple of $p$ or a multiple of $q$. Given this, it is not possible for a subset of columns to add to a multiple of $pq=Q$, and thus the number of independent columns cannot be decreased and so we cannot introduce any artifact errors from altering the underlying local-dimension from $\mathbb{Z}_q$ to $\mathbb{Z}_{pq}$.

Lastly, if $Q$ is no longer the product of merely two primes, while still being square-free, the same argument holds so long as all prime factors of $Q/q$ are larger than $p^*$.
\end{proof}
\fi

Composite local-dimension values being a ring and moreover a non-integral domain, provides for a somewhat atypical behavior. As we saw above with our two construction methods it is possible for the generators for the code to generate encodings of true composite local-dimension registers, but also we are able to effectively encode registers which are factors of the local-dimension value. The former of these is the typical scenario, whilst the second is a direct result of algebraic \textit{torsion}. We provide a more rigorous statement of how and when these result in the following lemma.

\begin{lemma}
Let $\mathcal{S}$ be a stabilizer code with local-dimension given by a ring $\mathcal{R}$. Further, let $\mathcal{L}$ be a generating set of logical operators, which are the kernels of the symplectic module $\mathcal{S}$ with the code quotiented out. Any member of $\mathcal{L}$ that has additive order not equal to the full size of $\mathcal{R}$ generates the torsion for the code. The additive order of each logical operator determines the logical dimension of that subspace.
\end{lemma}

This result is the definition of torsion for modules with the symplectic product. In a way this is the foil of Lemma \ref{ldim}. In that Lemma the logical dimension is determined by the order of the generators, whereas here it is determined by the order of the logical operators. These are equivalent, however, in the infinite case one may be easier than the other.


While the above is most direct for composite local-dimensions, the result also holds for factors of the reals module a value ($\mathbb{R}_p$ for some real number $p$). For instance, if one has a CSS code with mixed local-dimensions $\mathbb{Z}\times \mathbb{R}_{2\pi}$ and the checks on the circle group domain are all nontrivial divisors of $2\pi$, considered over the reals, then a torsion logical group is generated and the continuous logical register that would result normally is transformed into a discrete logical register. This fact was first noted in the work \cite{vuillot2024homological} and further extended to composite CSS codes in \cite{novak2024homological}, although in both a more homological approach was taken, while here we take a ring theoretic approach, enabling analysis of stabilizer codes beyond those that are CSS in form. In a sense, one can consider the GKP encoding \cite{gottesman2001encoding} as having (only) a torsion group, albeit with a single stabilized qubit where the operation generates a torsion group over the tesselation of $\mathbb{R}_{2\pi}$.

Another variety of a commutative ring which is non-integral-domain is the ring of multiples of a number. We show that such codes also have their parameters preserved. In what follows we use the notation $m\mathbb{Z}_q$ to indicate multiples of $m$ multiplied by a number module $q$.

\begin{lemma}
Let $m\mathbb{Z}$ indicate the integral multiples of $m$, and let $C$ be a stabilizer code with parameters $[[n,k,d]]_q$. Then we may construct a new code with parameters $[[n,k,d']]_{m\mathbb{Z}}$, with $d'\geq d$, using the LDI construction.
\end{lemma}

This result is immediate as all artifact errors from the code with local-dimension $q$ will have nonzero syndromes over $\mathbb{Z}$ and hence nonzero syndromes over $m\mathbb{Z}$, so the distance will remain , at least as large. The result also holds for $m\mathbb{Z}_p$ for $p>p^*$, which is isomorphic to the cyclic group over $p$ elements. Such local-dimensions can be of particular use for mode loss protection in systems such as photonic formats, other bosonic forms, and high-spin atoms with Zeeman dominated splitting \cite{michael2016new,gross2021designing,omanakuttan2024fault}. Such systems often suffer from noise which acts locally within the energy spectrum (between states with similar energies) due to Fermi's golden rule \cite{fermi1950nuclear,griffiths2019introduction}. This encoding, upon passing the leading number to the excitations in the codewords, would protect against up to $\lfloor\frac{m-1}{2}\rfloor$ excitation losses, acting as a Binomial encoding underneath a more traditional stabilizer code \cite{michael2016new}. 


\section{Mixed case}

In the prior section we showed one construction for true composite local-dimension codes by transforming them into an LDI form. Here we provide another method for constructing them by taking codes for local-dimension factors and protecting each factor against errors. To this end we begin with the following Lemma, which can be of use more generally for preserving symplectic relations in subspaces.



\begin{lemma}[Pick-and-mix]
Let the local-dimension of a system be given by a square-free number $Q=\prod_i p_i$. Suppose we have a set of Pauli operators over each local-dimension $p_i$, each set denoted $\mathcal{P}_i$, and which satisfy a given symplectic product matrix $[\odot(\mathcal{P}_i)]$. Then we may construct new Pauli operators $\mathcal{P}_i'$ such that the symplectic product values within each set are multiplied by a scalar, but between sets all operators commute. 
\end{lemma}

\begin{proof}
The entries in $[\odot(\mathcal{P}_i)]$ take values in $\mathbb{Z}_{p_i}$. We wish to keep the ratio between entries a constant, while ensuring that operators from $\mathcal{P}_j$ have a zero symplectic product over $Q$. Multiplying the operators in $\mathcal{P}_i$ with $Q/p_i$ turns the entries in the $\phi$ representation of each operator into values in $(Q/p_i)\mathbb{Z}_{p_i}$. We set this as our new $\mathcal{P}_i'$. Firstly this retains the ratio of the symplectic product, as well as the additive order of the elements. Secondly, consider the symplectic product between Paulis in the updated sets $\mathcal{P}_i'$ and $\mathcal{P}_j'$, which are given by values in $(Q/p_i)(Q/p_j)\mathbb{Z}_{p_i}\mathbb{Z}_{p_j}$, which notably must always be a multiple of $Q$ and so these operators commute. Then our new sets of operators are given by, abusing notation a little, $\mathcal{P}_i'=(Q/p_i)\mathcal{P}_i$.
\end{proof}

This generalized result is of somewhat limited use here as stabilizer codes require commuting operators, but this can easily be restricted to the case of each Pauli operator set consisting of solely commuting Pauli operators in their respective local-dimension. We may use the prior lemma to leverage codes from different local-dimensions to protect registers within different local-dimension factors. This then provides our second construction method for composite local-dimension codes:

\begin{lemma}
Let $Q=\prod_i p_i$ be a composite number and for each prime $p_i$ let $\mathcal{S}^{(p_i)}$ be a stabilizer code with parameters $[[n,k^{(p_i)},d^{(p_i)}]]_{p_i}$. Then we may construct a stabilizer code with parameters $[[n,\log_Q (Q^n/(\prod_i p_i^{n-k^{(p_i)}})),\min_i d^{(p_i)}]]_Q$.
\end{lemma}

In the above the superscript $(p_i)$ is an index indicating variable. This follows from the pick-and-mix lemma, lemma \ref{ldim}, and taking the minimal distance. Additionally in construction all codes are to be the same length, $n$, however, one could repeat codes of shorter length as well to add up to the same total length. This construction effectively partitions the logical space into registers of prime factors, protecting against errors of the same order as that of the original code.

For instance, for a $Q=6$ system with $n=35$ one could use the traditional five qubit code repeated seven times (using identity on the other registers each time) along with a qutrit (LDI) Steane code repeated five times. In this case the stabilizer group has a total of $(2^4)^7(3^6)^5$ elements, forming a logical space of size $6^{35}/(2^{28}3^{30})$. This subspace of size $2^7 3^5$ corresponds to $7$ qubits and $5$ qutrits, however, their logical spaces are connected. The logical operators are blocks of five $X^3$ gates, corresponding to the logical $X$ operators within the five qubit code, and blocks of seven $X^2$ gates, corresponding to the logical $X$ operators within the Steane-like code. The logical $Z$ operators are defined likewise. These logical operators manifestly have additive orders $2$ and $3$ for the five register and Steane-like codes respectively. We also remark that this corresponds to $5+2\log_6(2)$ logical $Q=6$ registers. If one were to add two of the logical operators for the five qubit codes in this would become $5$ logical registers flat. 

We can also mix codes, upon being put into an LDI form, with local-dimensions such as protecting against bit-flips with $\mathbb{Z}$ local-dimension and phase-flips with $\mathbb{R}_{2\pi}$ local-dimension--sometimes called rotor codes.






\begin{example}
    As another example we consider the following stabilizer code with local-dimension $Q=58$ and $n=28$. The generators for this code are four copies of Steane codes with nonzero powers being $29$-th power of the $58$-level $X$ and $Z$ operators. Each Steane code has transversal Clifford. We also add the generators of the quantum Reed-Muller code with local-dimension $29$, with the operators being $2$-nd powers of the $58$-level $X$ and $Z$ operators \cite{campbell2012magic}. This has a transversal $29$-level analogue of the $T$ gate, along with qudit $SUM$ gate. This encodes a single logical $29$-level register and four qubit registers, with a distance of $2$ (due to the Reed-Muller code). While this may seem to have a transversal and universal gate set, unfortunately these gates are not universal over $58$-level systems. It may be possible to leverage code switching to generate universality, however, we leave that as a future direction.
\end{example}

Here we explored composite local-dimension quantum stabilizer codes, but as future work careful consideration of any differences from composite local-dimensions with square factors would be useful, which may result from unifying prime-power research and the composite research herein. Another future direction includes extending the results to all varieties of commutative rings, although solving cases where there are practical applications are most important.

\section*{Acknowledgments}

We thank Andrew Jena for useful comments which helped with expanding the conditions on Theorem \ref{pdist} and Victor Albert for pointing out the reference \cite{novak2024homological}.

\if{false}
\clearpage

\section{Submission of Papers for Review}

Papers in the form of a PDF file, formatted as described below, may be
submitted online as indicated on the webpage. 

The deadline for registering the paper and uploading the manuscript is \textbf{January 22, 2025} (anywhere on earth). \textbf{No extensions will be given.}

A paper's primary content is restricted in length to \textbf{5 pages}, but
authors are allowed an optional 6th page only containing references. The submission may contain a link to a longer, online version of the submitted manuscript, if the authors wish to include one.
 
The IEEEtran-conference style should be used as presented here. Submissions should use a font size of at least 10 points and have reasonable margins on all the 4 sides of the text.

Papers that are eligible for the student paper award must  include the comment ``THIS PAPER IS ELIGIBLE FOR THE STUDENT PAPER AWARD." as a first line in the abstract.

\section{Submission of Accepted Papers}

Accepted papers will be published in full. For papers that are eligible for the student paper award, please do not forget to remove the comment regarding the paper's eligibility from the abstract.

\section{Paper Format}

\subsection{Templates}

The paper (A4 or letter size, double-column format, not exceeding 5 pages plus an optional 6th page only containing references) should be formatted as shown in this sample \LaTeX{} file
\cite{Laport:LaTeX, GMS:LaTeXComp, oetiker_latex, typesetmoser}.

The use of Microsoft Word or other text processing systems instead of \LaTeX{} is strongly
discouraged. Users of such systems should attempt to duplicate the
style of this example, in particular the sizes and type of font, as
closely as possible.

\subsection{Formatting}

The style of references, equations, figures, tables, etc. should be
the same as for the \emph{IEEE Transactions on Information
  Theory}. 
  
  The source file of this template paper contains many more
instructions on how to format your paper. For example, example code for
different numbers of authors, for figures and tables, and references
can be found (they are commented out).

For instructions on how to typeset math, in particular for equation
arrays with broken equations, we refer to \cite{typesetmoser}.

Final papers should have no page numbers and no headers or footers (both will be added during the production of the proceedings).
The top and bottom margins should be at least 0.5 inches to leave room for page numbers.
All fonts should be embedded in the pdf file.

\subsection{PDF Requirements}

Only electronic submissions in form of a PDF file will be
accepted. The PDF file has to be PDF/A compliant. A common problem is
missing fonts. Make sure that all fonts are embedded. (In some cases,
printing a PDF to a PostScript file, and then creating a new PDF with
Acrobat Distiller, may do the trick.) More information (including
suitable Acrobat Distiller Settings) is available from the IEEE
website \cite{IEEE:pdfsettings, IEEE:AuthorToolbox}.

\section{Conclusion}

We conclude by pointing out that on the last page the columns need to
balanced. Instructions for that purpose are given in the source file (but are commented out).

Moreover, example code for an appendix (or appendices) can also be
found in the source file (they are commented out).


\section*{Acknowledgment}

We are indebted to Michael Shell for maintaining and improving
\texttt{IEEEtran.cls}. 

\fi

\bibliographystyle{IEEEtran}
\bibliography{main}

\end{document}